\documentclass[journal]{IEEEtran}

\usepackage[numbers,sort&compress,square]{natbib}   
\usepackage{graphicx}
\usepackage{subfigure}
\usepackage{amsmath,amssymb}
\usepackage{amsfonts}
\usepackage{array}
\usepackage{amsthm}
\usepackage{mathrsfs}
\usepackage{setspace}
\usepackage{color}
\usepackage[ruled,linesnumbered]{algorithm2e}
\usepackage{soul}

\newtheorem{theorem}{Theorem}

\newtheorem{lemma}[theorem]{Lemma}

\theoremstyle{definition}

\begin{document}

	\title{\huge Efficient Broadcast for Timely Updates in Mobile Networks}
	
	%
	
	\author{\IEEEauthorblockN{Yu-Pin Hsu\vspace{-.2cm}}
		
		\thanks{Y.-P. Hsu is with the Department of Communication Engineering, National Taipei University, Taiwan (Email: yupinhsu@mail.ntpu.edu.tw).  This work was supported by MOST, Taiwan, under 107-2221-E-305-007-MY3.} 
		
	}

	\maketitle
	
	
	\begin{abstract}
		This study considers a wireless network where an access point (AP) \textit{broadcasts} timely updates to numerous mobile users. The timeliness of information owned by a user is characterized by the \textit{age of information}. Frequently broadcasting the timely updates at constant maximum power can minimize the age of information for all users, but wastes valuable communication resources (ie., time and energy). To address the age-energy trade-off, it is critical to develop an efficient scheduling algorithm that identifies broadcast times and allocates power. Moreover, unpredictable user movement would cause rapidly varying communication channels; in particular, those channels can be \textit{non-stationary}. Our main contribution is to develop an \textit{online scheduling algorithm} and a \textit{channel-agnostic scheduling algorithm} for such a mobile network with a provable performance guarantee.
	\end{abstract}
	

	\section{Introduction}
	In recent years, there has been an explosive growth of {real-time} applications for mobile users that rely on timely information. For example, smart parking applications \cite{gu2019timely} require (downloading) timely parking information. In such applications, the timeliness of the information is critical. Thus, the \textit{age of information} was recently proposed in \cite{age:kaul} as a metric to measure the information timeliness. The goal is to develop a network to minimize the age of information.

	In this study, we consider a scenario where several  users run applications that request the same information (e.g., parking information) simultaneously through an access point (AP). To serve their requests, the AP  sends timely updates to those users from time to time. Moreover, because of the broadcast nature of the wireless medium, the AP can simultaneously transmit an update to multiple users by \textit{broadcast} (with a single transmission). To minimize the age of information, the AP transmits as many latest updates as possible and transmits them with the highest power level (such that most users can successfully decode the updates).  
	However, as the demand for mobile services continues to increase significantly, communication resources (i.e., time and energy) at the AP become more valuable. 
	To fully exploit the precious resources, this study aims to develop scheduling algorithms that determine \textit{when to transmit an update} and \textit{how to allocate the power for each transmission}. The goal is to strike a balance between the age of information  and energy consumption.

	The scheduling design 
	for managing the age-energy trade-off has drawn significant attention, e.g., \cite{nath2018optimum,arafa2019age,gu2019timely,tseng2019online}. Nearly all prior research assumed \textit{stochastic} environments where the environmental variables (e.g., communication channels) follow some probability distribution or have some stationary assumptions. However, a mobile user can move at will. Because of the users' highly unpredictable movement, it is infeasible to assume stationary channels for our problem. Our previous work \cite{tseng2019online} was an initial attempt to investigate the trade-off in an \textit{adversarial} environment. However, \cite{tseng2019online} only considered a \textit{single} user in the \textit{ON-OFF} channel caused by a \textit{fixed} power manager. In contrast, this study considers \textit{multiple} mobile users in general fading channels (with \textit{more than two states}). The significant uncertainty caused by multiple mobile users poses a major challenge. Moreover, this study allows adaptive power selection. This additional option can save more energy, but further complicates our scheduling design.

	The primary contribution of this study is to develop \textit{online} and \textit{channel-agnostic} scheduling algorithms. While the former needs the present channel states only (without future information), the latter needs no  channel state. We show that both algorithms have a \textit{universal} performance guarantee that is independent of the number of users and their movement.

	\section{System Overview} \label{section:system}
	
	\subsection{Network Model} \label{subsection:network}
	
	We consider a wireless network consisting of an AP and $N$ mobile users, where users $1, \cdots, N$ move in the AP service area. The users run applications that request the same timely updates through the AP. Divide time into slots $1, 2, \cdots, T$, where~$T$ is the time horizon under consideration. 

	
	At the beginning of each slot, the AP \textit{decides} whether or not to serve the users. If the AP decides to transmit for a slot, then it obtains an update from an information source, allocates the transmission power, and transmits the update by \textit{broadcast} during that slot. 
	In this study, we assume that the AP can immediately obtain the update from the information source (through wired networks) and focus on the bottleneck between the AP and users (through wireless networks). 
	Let $d(t)\in \{0, 1, \cdots, M\}$ be the AP's decision\footnote{To inform the users, the AP can broadcast its scheduling decision using the control channel \cite[Chapter 19]{dahlman20164g} at the beginning of each slot.} for slot~$t$, where $d(t)=0$ if it decides not to serve and $d(t) \in \{1, \cdots, M\}$ is the power level allocated for broadcast in slot~$t$ if it decides to serve. Suppose that the power consumption increases with the power level index.

	%
	
	%
	
	%

	We use $\mathbf{1}_{i,k}(t) \in \{0, 1\}$ to indicate if user~$i$ can receive (i.e., successfully decode) the update in slot~$t$ if the AP takes power level~$k$, where $\mathbf{1}_{i,k}(t)=1$ if it can but $\mathbf{1}_{i,k}(t)=0$ if it cannot. The indicator function $\mathbf{1}_{i,k}(t)$ depends on the channel quality\footnote{\label{footnote:channel} Consider the slow fading  model in \cite[Chapter 5]{tse2005fundamentals} as an example. Let $h_i(t)$ be the channel gain for user~$i$ in slot~$t$. Let $\mathbf{SNR}_{i,k}(t)$ be the signal-to-noise (SNR) ratio  for user~$i$ under power level~$k$ in slot~$t$. Then, if $\log (1+|h_i(t)|^2\mathbf{SNR}_{i,k}(t))$	is greater than a threshold such that the update can be reliably delivered during the slot period, then $\mathbf{1}_{i,k}(t)=1$; otherwise, $\mathbf{1}_{i,k}(t)=0$.} between the AP and user~$i$ in slot~$t$. 
	We suppose that, for each slot, the AP can deliver an update to all users with the maximum power level~$M$, that is, the maximum power level~$M$ specifies the AP service area. 
	
	Let $s_i(t)=(\mathbf{1}_{i,1}(t), \cdots, \mathbf{1}_{i,M}(t))$ be the \textit{channel state} of user~$i$ in slot~$t$. Note that if $\mathbf{1}_{i,k}(t)=1$, then $\mathbf{1}_{i, k'}(t)=1$ for all $k' \geq k$. Thus, there are $M$ potential states\footnote{Follow Footnote~\ref{footnote:channel}. We discretize the channel gain and the SNR into $M$ bins according to the $M$ power levels. The channel gain and SNR such that the AP cannot reliably deliver with power level~$k-1$ but can with power level~$k$ are associated with state $(0, \cdots, 0, 1, \cdots, 1)$ with the first $k-1$ elements being zeros.} for $s_i(t)$, i.e., $(1, \cdots, 1)$, $(0, 1, \cdots, 1)$, $\cdots$, $(0, \cdots, 0, 1)$. Let $\mathbf{s}_i=(s_i(1), \cdots, s_i(T))$ be the \textit{channel state pattern} of user~$i$ over slots, and represent the channel states caused by its movement. Because of  unpredictable movement, the channel state pattern~$\mathbf{s}_i$ can be \textit{arbitrary} with no stationary\footnote{Follow Footnote~\ref{footnote:channel}. This study does not assume any distribution for $h_i(t)$ (e.g., Rayleigh distribution).} property  for all~$i$. Moreover,~$\mathbf{s}_i$ and~$\mathbf{s}_j$ for different~$i$ and~$j$ can have any correlation, for example, when a group of users move together. Let $\mathbf{s}=(\mathbf{s}_1, \cdots, \mathbf{s}_N)$ represent the \textit{entire channel state pattern} of all users.

	\subsection{Age Model} \label{subsection:age model}
	We use the age of information \cite{age:kaul} to measure the information timeliness for each user. If a user receives an update in a slot, then its age of information becomes zero at the end of that slot; otherwise, the age of information at the end of that slot increases by one from the previous slot. Let $a_i(t)$ be the age of information for user~$i$ at the end of slot~$t$. Then,  we can describe age $a_i(t)$ for user~$i$ in slot~$t$ by
	\begin{align}
		a_i(t)=\left\{
		\begin{array}{ll}
			0 & \text{if $\mathbf{1}_{i,d(t)}(t)=1$;}\\
			a_i(t-1)+1 & \text{if $\mathbf{1}_{i,d(t)}(t)=0$.}
		\end{array}
		\right.
		\label{eq:age-dynamic}
	\end{align}
	We assume the initial age $a_i(0)=0$ for all $i$. 
	

	\subsection{Problem Formulation} \label{subsection:problem}

	A \textit{scheduling algorithm} $\pi=\{d(1), \cdots, d(T)\}$ specifies decision $d(t)$ for all slots. To develop scheduling algorithms striking a balance between the users' information timeliness and the AP's energy consumption, we define an \textit{age cost} and a \textit{transmission cost} as follows: Suppose that the age of one unit in a slot incurs a cost of one unit in that slot. Then, the age cost incurred by the stale information at user~$i$ in slot~$t$ is $a_i(t)$. We consider the average age cost $(\sum_{i=1}^N a_i(t))/N$ over the number of users (for a fair comparison across different numbers $N$ of users). Moreover, suppose that the power level~$k$ incurs a transmission cost of $C_k$ units. For example, $C_k=\mathcal{W}\cdot \mathcal{E}(k)$, where $\mathcal{W}$ is the weight between the unit energy consumption and the unit age, and $\mathcal{E}(k)$ is the energy consumption under power level~$k$. Note that the function $C_k$ in this study can be any non-decreasing function with power level~$k$, that is, it does not need to linearly increase with $\mathcal{E}(k)$.

	For an entire channel state pattern $\mathbf{s}$, we define the \textit{total cost} $J(\mathbf{s},\pi)$ under a scheduling algorithm $\pi$ by the sum of the transmission costs and the average age costs:
	\begin{align}
		J(\mathbf{s}, \pi) =\sum_{t=1}^T \Bigl(C_{d(t)}+ \frac{1}{N} \sum_{i =1}^N a_i(t)\Bigr),
		\label{eq:cost}
	\end{align}
	where we set $C_0=0$ indicating the zero transmission cost when $d(t)=0$. 
	Eq.~(\ref{eq:cost}) captures the age-energy trade-off. The higher power level~$d(t)$ the AP takes (in the first term of Eq.~(\ref{eq:cost})), the more users can receive the latest update, yielding a smaller average age cost (in the second term of Eq.~(\ref{eq:age-dynamic})). Conversely, taking a lower power level increases the age of information for those users who cannot decode the update as a result of poor channels. 
	
	
	This study proposes two scheduling algorithms. On the one hand, for the case where the AP can obtain channel state~$s_i(t)$ for all~$i$ in slot~$t$  (by channel estimation techniques such as \cite{she2020tutorial}), this study develops an \textit{online} scheduling algorithm that makes decision~$d(t)$ for slot~$t$ with the \textit{present} channel states $s_i(t)$ (with no \textit{future} channel state information) and the \textit{present} ages~$a_i(t)$ of all users~$i$. On the other hand, for the case where the AP cannot obtain the channel state information, this study develops a \textit{channel-agnostic} scheduling algorithm that makes decision $d(t)$ for slot~$t$ with $C_1$, $C_M$ (with no channel state information), and the \textit{present} ages.

	We analyze the proposed algorithms in terms of \textit{competitiveness} against an optimal \textit{offline} scheduling algorithm (which has the entire channel state pattern~$\mathbf{s}$ along with the time horizon~$T$ as  prior information). To that end, for an entire channel state pattern $\mathbf{s}$, let $OPT(\mathbf{s})=\min_{\pi'} J(\mathbf{s},\pi')$ be the minimum total cost for all possible (offline) scheduling algorithms~$\pi'$. If there is a constant~$\gamma$ such that $J(\mathbf{s}, \pi) \leq \gamma \cdot OPT(\mathbf{s})$, 
	for all possible entire channel state patterns~$\mathbf{s}$, then the constant~$\gamma$ is called the \textit{competitive ratio} of the scheduling algorithm~$\pi$. 
	

	\section{Scheduling Algorithm Design}\label{section:scheduling}
	This study approaches the scheduling problem by leveraging \textit{primal-dual techniques} \cite{online-compuatation-naor} for \textit{linear} programs. To that end, Section~\ref{subsection:qeuue} constructs a virtual queueing system for describing the age evolution in Eq.~(\ref{eq:age-dynamic}). With the assistance of the virtual queueing system, Section~\ref{subsection:primal-dual} proposes a linear program whose optimal objective value is the lower bound on the total cost. In addition, Section~\ref{subsection:primal-dual} formulates the dual program of the linear program (as a primal program) for analyzing our proposed algorithms (using duality theory). Then, Section~\ref{subsection:pda} develops a primal-dual algorithm that \textit{online} produces a feasible solution to the primal program and the dual program . Next, Section~\ref{subsection:online} proposes a (randomized) online scheduling algorithm by casting the fractional solution produced by the primal-dual algorithm in each slot to a randomized decision for that slot. Finally, Section~\ref{subsection:channe-agnostic} proposes a (randomized) channel-agnostic scheduling algorithm that achieves the same competitive ratio as  the online scheduling algorithm does. 
	
	\subsection{Virtual Queueing System} \label{subsection:qeuue}
	This section constructs a virtual queueing system consisting of queues $1, \cdots, N$ (corresponding to users $1, \cdots, N$) so that the queue size evolution is equivalent to the age evolution in Eq.~(\ref{eq:age-dynamic}). 
	
	At the beginning of each slot~$j$, \textit{each} queue $i$ has a newly arriving packet~$j$. Note that a total of $N$ packet~$j$'s arrive at the queueing system in slot~$j$ because the system has $N$ queues. 
	In addition, for each slot~$j$, if user~$i$ can receive the update (in the real-world mobile network), then queue~$i$ flushes all its packets (in the virtual queueing system); otherwise, queue~$i$ idles.
	%
	
	According to the arrival and service processes of queue~$i$, the queue size at the end of slot~$t$ becomes zero if $\mathbf{1}_{i,d(t)}(t)=1$, or increases by one if $\mathbf{1}_{i,d(t)}(t)=0$. The queue size evolution is identical to the age evolution in Eq.~(\ref{eq:age-dynamic}). Thus, the queue size of queue~$i$ at the end of slot~$t$ exactly expresses~$a_i(t)$.
	
	\subsection{Primal-Dual Formulation} \label{subsection:primal-dual}
	Leveraging the queueing system constructed in Section~\ref{subsection:qeuue}, this section proposes an integer program for \textit{optimally} solving the offline scheduling problem. Let $z_{i,j}(t)\in \{0, 1\}$ indicate if packet~$j$ (arriving in slot~$j$) stays at queue~$i$ at the end of slot~$t$, where $z_{i,j}(t)=1$ if it does but \mbox{$z_{i,j}(t)=0$} otherwise. According to Section~\ref{subsection:qeuue}, age $a_i(t)$ in slot~$t$ is the queue size of queue~$i$ at the end of slot~$t$. Thus, age $a_i(t)$ in Eq.~(\ref{eq:age-dynamic}) can be expressed by $a_i(t)=\sum_{j=1}^t z_{i,j}(t)$, counting all packets arriving at queue~$i$ by slot~$t$. Moreover, let $x_k(t) \in \{0,1\}$ indicate if the power level~$k$ is selected in slot~$t$, where $x_k(t)=1$ and $x_{k'}(t)=0$ for all $k' \neq k$ if and only if $d(t)=k$.
	Then, the cost $C_{d(t)}$ in Eq.~(\ref{eq:cost}) can be expressed as $C_{d(t)}=\sum_{k=1}^M C_k x_k(t)$. Substituting $C_{d(t)}$ and $a_i(t)$ in Eq.~(\ref{eq:cost}) by the new expressions, we can re-write the total cost $J(\mathbf{s}, \pi)$ by a linear function in terms of $x_k(t)$ and $z_{i,j}(t)$: 
	\begin{align}
		J(\mathbf{s}, \pi)=\sum_{t=1}^T \left(\sum_{k=1}^M C_k x_k(t) +\frac{1}{N} \sum_{i=1}^N\sum_{j=1}^t z_{i,j}(t)\right). \label{eq:cost-2}
	\end{align}

	Then, we propose the following \textit{integer program} for optimally solving our scheduling problem when the entire channel state pattern~$\mathbf{s}$ is known in advance:
	\begin{subequations}	\label{integer-program}
		\begin{align}
			\min & \hspace{.2cm}\sum_{t=1}^T \left(\sum_{k=1}^M C_k x_k(t)  + \frac{1}{N}\sum_{i=1}^N\sum_{j=1}^t z_{i,j}(t)\right) \label{inter-program:objective}\\
			\text{s.t.} & \hspace{.2cm} z_{i,j}(t) + \sum_{\tau =j}^{t} \sum_{k=1}^M \mathbf{1}_{i,k}(\tau) x_k(\tau) \geq 1 \nonumber\\
			&\hspace{.2cm}\text{for all $i=1, \cdots, N$, $j=1, \cdots, t$, and $t=1, \cdots, T$;}	\label{interger-program:const1}\\
			& \hspace{.2cm} \text{$x_k(t), z_{i,j}(t) \in \{0,1\}$ for all~$i$,~$j$,~$k$,~$t$.} \label{interger-program:const2} 
		\end{align}
	\end{subequations}
	The constraint in Eq.~(\ref{interger-program:const1}) guarantees that, for each slot~$t$, each packet~$j \leq t$ (arriving at queue~$i$ by slot~$t$) \textit{either} remains in queue~$i$ in slot~$t$ (i.e., $z_{i,j}=1$ in the first term of Eq.~(\ref{interger-program:const1})) \textit{or} has been flushed by slot~$t$ (i.e., there exists a prior slot~$\tau=j, \cdots, t$ and a power level~$k=1, \cdots, M$ such that $x_k(\tau)=1$ and $\mathbf{1}_{i,k}(\tau)=1$ in the second term of Eq.~(\ref{interger-program:const1}).


By relaxing the integral constraints in Eq.~(\ref{interger-program:const2}) to real numbers, we have a \textit{linear program}. Because of the relaxation, the minimum objective value for the linear program is the lower bound on the minimum total cost (obtained by the integer program). Moreover, unlike the integer program, a feasible solution to the linear program can be fractional, which no longer represents an immediate decision to broadcast an update or allocate power. Subsequently, Section~\ref{subsection:online} will cast a fractional solution for~$x_k(t)$ to a \textit{probabilistic} decision in slot~$t$.

%

To analyze the solution produced by the proposed algorithms, we leverage duality theory \cite{online-compuatation-naor}. Thus, we refer to the linear program as a \textit{primal program} and formulate its \textit{dual program} as follows: 
\begin{subequations}	\label{dual-program}
	\begin{align}
		\max& \hspace{.2cm}\sum_{t=1}^T \sum_{i=1}^N \sum_{j=1}^t y_{i,j}(t) \label{dual-program:objective}\\
		\text{s.t.} & \hspace{.2cm}\sum_{i=1}^N\mathbf{1}_{i,k}(t) \sum_{j=1}^t\sum_{\tau=t}^{T} y_{i,j}(\tau) \leq C_k \text{\,\,for all $k$, $t$;} 
		\label{dual-program:const1}\\
		& \hspace{.2cm}\text{$0 \leq y_{i,j}(t) \leq \frac{1}{N}$ for all $i$,~$j$,~$t$.} \label{dual-program:const2} 
	\end{align}
\end{subequations}

\subsection{Primal-Dual Algorithm}\label{subsection:pda}

\begin{algorithm}[t]
	\SetAlgoLined 
	\SetKwFunction{Union}{Union}\SetKwFunction{FindCompress}{FindCompress} \SetKwInOut{Input}{input}\SetKwInOut{Output}{output}
	%

	$x_k(t)$, $z_{i,j}(t)$, $y_{i,j}(t)$  $\leftarrow 0$ for all~$i$,~$j$,~$k$,~$t$\; \label{primal-dual-alg1:initial}
	
	$\theta \leftarrow (1+\frac{1}{C_{M}})^{\lfloor C_{1} \rfloor}-1$\;	\label{primal-dual-alg1:constant}

	\tcc{For each  slot~$t$, update as follows:}

	Identify a minimum~$k$ (denoted by $k^*_t$) such that $\mathbf{1}_{i,k}(t)=1$ for all $i=1, \cdots, N$\; \label{primal-dual-alg1:max}
	\For{$j=1$ \KwTo~$t$\label{primal-dual-alg1:for1}}{	
		\If{$\sum_{\tau=j}^{t}x_{k^*_{\tau}}(\tau)<1$  \label{primal-dual-alg1:condition}}{			
			
			$z_{i,j}(t) \leftarrow 1- \sum_{\tau=j}^{t}x_{k^*_{\tau}}(\tau)$ for all $i=1, \cdots, N$\; \label{primal-dual-alg1:z1}
			$x_{k^*_t}(t) \leftarrow x_{k^*_t}(t)+ \frac{1}{C_{k^*_t}}\sum_{\tau=j}^{t} x_{k^*_{\tau}}(\tau)+\frac{1}{\theta \cdot C_{k^*_t}}$\; 	 \label{primal-dual-alg1:x}
			$y_{i,j}(t)\leftarrow \frac{1}{N}$ for all $i=1, \cdots, N$\; \label{primal-dual-alg1:y1}
			
		}\label{primal-dual-alg1:for1end}
	}

	\caption{Primal-dual algorithm.}
	\label{primal-dual-alg1}
\end{algorithm}

This section proposes a primal-dual algorithm in Alg.~\ref{primal-dual-alg1} for obtaining a feasible solution to the primal and dual programs in the \textit{online} setting, where for each new slot~$t$ Alg.~\ref{primal-dual-alg1} has the constraints in Eqs.~(\ref{interger-program:const1}) and~(\ref{dual-program:const1}) only until slot~$t$ (but has no entire set of constraints). 

Alg.~\ref{primal-dual-alg1} initializes all the variables to zero in Line~\ref{primal-dual-alg1:initial}. For each new slot~$t$, Alg.~\ref{primal-dual-alg1} updates all the variables according to the present channel state~$s_i(t)$ for all~$i$ (with no future channel state pattern). Line~\ref{primal-dual-alg1:max} identifies a maximum value (denoted by $k^*_t$ for slot~$t$) of~$k$ such that $\mathbf{1}_{i,k}(t)=1$ for all~$i$. The underlying idea is that for each slot~$t$, our scheduling algorithm (in Section~\ref{subsection:online}) decides either $d(t)=k^*_t$ \textit{with some probability} or $d(t)=0$ otherwise. The decision $d(t)=k^*_t$ allocates the minimum power for successfully broadcasting to \textit{all} users in slot~$t$; in turn, all queues flush their packets in slot~$t$.

This probabilistic decision is based on the present value of $x_{k^*_t}(t)$ updated in Line~\ref{primal-dual-alg1:x}. For each slot,~$t$, Alg.~\ref{primal-dual-alg1} updates the value of~$x_{k^*_t}(t)$ by iteration (in Line~\ref{primal-dual-alg1:for1}) from iteration $j=1$ (for the first arriving packet~1) until iteration $j=t$ (for the newly arriving packet~$t$), if the condition in Line~\ref{primal-dual-alg1:condition} holds. To understand the idea behind the condition in Line~\ref{primal-dual-alg1:condition} and the update in Line~\ref{primal-dual-alg1:x}, we interpret the value of $\sum_{\tau=j}^t\sum_{k=1}^M \mathbf{1}_{i,k}(\tau)x_k(\tau)$ (in Eq.~(\ref{interger-program:const1})) as the cumulative probability that packet~$j$ gets flushed by slot~$t$ since its arrival (in slot~$j$). Note that, for each slot~$\tau$, Alg.~\ref{primal-dual-alg1} updates the value of $x_{k^*_{\tau}}(\tau)$ only (in line~\ref{primal-dual-alg1:x}) but keeps the value of $x_k(\tau)$ for all $k \neq k^*_{\tau}$ unchanged. In addition, note that $\mathbf{1}_{i,k^*_{\tau}}(\tau)=1$ for all~$\tau$. 
Thus, the value of $\sum_{\tau=j}^t\sum_{k=1}^M \mathbf{1}_{i,k}(\tau)x_k(\tau)$ becomes that of $\sum_{\tau=j}^t x_{k^*_{\tau}}(\tau)$. In other words, the value of $\sum_{\tau=j}^t x_{k^*_{\tau}}(\tau)$ implies the cumulative flushing probability of packet~$j$ by slot~$t$. 

With the above interpretation, the condition in Line~\ref{primal-dual-alg1:condition} indicates if packet~$j$ has been flushed by slot~$t$. On the one hand, if $\sum^t_{\tau=j}x_{k^*_{\tau}}(\tau) \geq 1$, then packet~$j$ has been flushed; thus, no variable needs to be updated. On the other hand, if $\sum^t_{\tau=j}x_{k^*_{\tau}}(\tau) <1$, then packet~$j$ might still exist in slot~$t$; thus, its associated variables are updated. For each packet that may exist in slot~$t$, Line~\ref{primal-dual-alg1:x} increases the value of $x_{k^*_t}(t)$, that is, the more packets that exist, the higher the flushing probability.

Moreover, according to Line~\ref{primal-dual-alg1:x}, the cumulative flushing probability $\sum_{\tau=j}^t x_{k^*_{\tau}}(\tau)$ is updated by
\begin{align*}
	\sum_{\tau=j}^t x_{k^*_{\tau}}(\tau) \leftarrow \,\,&\sum_{\tau=j}^t x_{k^*_{\tau}}(\tau)+ \frac{1}{C_{k^*_t}} \sum_{\tau=j}^t x_{k^*_{\tau}}(\tau)+\frac{1}{\theta \cdot C_{k^*_t}}\\
	&=\left(1+\frac{1}{C_{k^*_t}}\right)\sum_{\tau=j}^t x_{k^*_{\tau}}(\tau)+\frac{1}{\theta \cdot C_{k^*_t}}.
\end{align*}
That is, Line~\ref{primal-dual-alg1:x} increases the cumulative flushing probability with the multiplicative scale of $(1+1/C_{k^*_t})$ and an additive scale of $1/(\theta \cdot C_{k^*_t})$. The term $C_{k^*_t}$ appears in the denominators because the higher the value of $C_{k^*_t}$ is, the less the flushing probability rises. Moreover, the constant value $\theta$ in Line~\ref{primal-dual-alg1:x} is specified in Line~\ref{primal-dual-alg1:constant} to satisfy the dual constraints in Eq.~(\ref{dual-program:const1}). We want to emphasize that, although there are $N$ queues that have packet~$j$, Line~\ref{primal-dual-alg1:x} increases the value of $x_{k^*_t}(t)$ only \textit{once} for iteration~$j$ (instead of intuitively increasing that value for all packet $j$'s at queues~$1, \cdots, N$). The underlying idea is that, with the decision $d(t)=k^*_t$ or $d(t)=0$ in our scheduling algorithm (in Section~\ref{subsection:online}), all queue sizes are \textit{the same} for all slots. Thus, Line~\ref{primal-dual-alg1:x} has been scaled to capture all packet~$j$'s for each iteration~$j$. Such a user-number-independent update for each iteration can reduce the computational complexity of our channel-agnostic algorithm (in Section~\ref{subsection:channe-agnostic}).

In addition, Line~\ref{primal-dual-alg1:z1} updates the value of $z_{i,j}(t)$ to that of $1-\sum^t_{\tau=j}x_{k^*_{\tau}}(\tau)$ to satisfy the primal constraints in Eq.~(\ref{interger-program:const1}). Finally, Line~\ref{primal-dual-alg1:y1} updates the value of $y_{i,j}(t)$ to $1/N$ to maximize the dual objective value in Eq.~(\ref{dual-program:objective}).

We analyze the primal objective value in Eq.~(\ref{inter-program:objective}) computed by Alg.~\ref{primal-dual-alg1} as follows (dy the duality theory).
\begin{theorem} \label{theroem:competitive-ratio1}
	The primal objective value in Eq.~(\ref{inter-program:objective}) computed by Alg.~\ref{primal-dual-alg1} at the end of slot~$T$ is bounded above by
	\begin{align*}
		\left( 1+ \frac{1}{(1+\frac{1}{C_{M}})^{\lfloor C_{1} \rfloor}-1}\right) OPT(\mathbf{s}),
	\end{align*}
	for all possible entire channel state pattern~$\mathbf{s}$.
\end{theorem}
\begin{proof}	
See Appendix~\ref{appendix:theroem:competitive-ratio1}.
\end{proof}

\subsection{Online Scheduling Algorithm}\label{subsection:online}
\begin{algorithm}[!t]
	\SetAlgoLined 
	\SetKwFunction{Union}{Union}\SetKwFunction{FindCompress}{FindCompress} \SetKwInOut{Input}{input}\SetKwInOut{Output}{output}
	%
	
	$x_{\text{pre-sum}}, x_{\text{sum}},  x_{k}(t)\leftarrow 0$  for all~$k$,~$t$\; \label{online-alg1:initial}
	
	$\theta \leftarrow (1+\frac{1}{C_M})^{\lfloor C_1 \rfloor}-1$\;
	Choose a uniformly random number $u \in [0,1)$\; \label{online-alg1:random}
	\tcc{For each slot~$t$, do as follows:}
	
	Identify a minimum~$k$ (denoted by $k^*_t$) such that $\mathbf{1}_{i,k}(t)=1$ for all $i=1, \cdots, N$\; \label{online-alg1:max}

	\For{$j=1$ \KwTo~$t$\label{online-alg1:for1}}{
		\If{$\sum_{\tau=j}^{t}x_{k^*_{\tau}}(\tau)<1$}{		
			$x_{k^*_t}(t) \leftarrow x_{k^*_t}(t)+ \frac{1}{C_{k^*_t}}\sum_{\tau=j}^{t} e(\tau)+\frac{1}{\theta \cdot C_{k^*_t}}$\; 	 \label{online-alg1:x}			
		}
		
	}
	$x_{\text{pre-sum}} \leftarrow x_{\text{sum}}$\;  \label{online-alg1:pre-sum}
	$x_{\text{sum}} \leftarrow x_{\text{sum}}+\min\{x_{k^*_t}(t),1\}$\; \label{online-alg1:sum}
	
	\uIf{$x_{\text{pre-sum}} \leq u < x_{\text{sum}}$\label{online-alg1:condition}}{
		$d(t) \leftarrow k^*_t$\; \label{online-alg1:tx}
		$u \leftarrow u+1$\;  \label{online-alg1:u+1}
	} 
	\Else{
		$d(t) \leftarrow 0$\; \label{online-alg1:idle}
	}
	\label{online-alg1:end}
	
	\caption{Online scheduling  algorithm.}
	\label{online-alg1}
\end{algorithm}

By leveraging the value of $x_{k^*_t}(t)$ produced by Alg.~\ref{primal-dual-alg1}, this section proposes a (randomized) online scheduling algorithm in Alg.~\ref{online-alg1}. Alg.~\ref{online-alg1} updates the variable $x_{k^*_t}(t)$ in Line~\ref{online-alg1:x} in the same manner as Alg.~\ref{primal-dual-alg1}. Moreover, Alg.~\ref{online-alg1} introduces additional variables $x_{\text{pre-sum}}$ and $x_{\text{sum}}$. The value of $x_{\text{pre-sum}}$ in slot~$t$ is the cumulative value of $\min\{x_{k^*_t}(t),1\}$ until $t-1$ (see Line~\ref{online-alg1:pre-sum}), the value of $x_{\text{sum}}$ in slot~$t$ is the cumulative value of $\min\{x_{k^*_t}(t),1\}$ until slot~$t$ (see Line~\ref{online-alg1:sum}). 

Alg.~\ref{online-alg1} selects a uniformly random number $u \in [0,1)$ in Line~\ref{online-alg1:random}. According to Lines~\ref{online-alg1:condition} and~\ref{online-alg1:u+1}, if there exists $k \in \mathbb{N}$ such that $u+k \in [x_{\text{pre-sum}}, x_{\text{sum}})$, then the AP decides $d(t)=k^*_t$ for slot~$t$ (see Line~\ref{online-alg1:tx}), that is, the AP allocates the minimum power such that all users can receive the update; otherwise, the AP decides $d(t)=0$ (see Line~\ref{online-alg1:idle}). The idea behind Alg.~\ref{online-alg1} is that, with the uniformly random choice of $u$, the probability of broadcasting to \textit{all} users (or the probability of flushing all queues) in slot~$t$ becomes $\min\{x_{k^*_t}(t),1\}$ and the cumulative probability of flushing packet~$j$ in slot~$t$ is $\min\{\sum_{\tau=j}^t x_{k^*_{\tau}}(\tau),1\}$. A similar idea was also used to generate random numbers for a given cumulative distribution function. 

We want to emphasize that always choosing $d(t)=k^*_t$ for each transmission (such that everyone can receive the update) may not be optimal. However, we can show that the \textit{expected} competitive ratio (over the randomness of~$u$) of our algorithm can be guaranteed as follows: 

\begin{theorem} \label{theorem:online-alg1}
	The expected competitive ratio of Alg.~\ref{online-alg1} is	
	\begin{align*}
		\left( 1+ \frac{1}{(1+\frac{1}{C_{M}})^{\lfloor C_{1} \rfloor}-1}\right).
	\end{align*}	
\end{theorem}
\begin{proof}
See Appendix~\ref{appendix:theorem:online-alg1}.
\end{proof}

The competitive ratio in Theorem~\ref{theorem:online-alg1} is independent of the number $N$ of users, the entire channel state pattern~$\mathbf{s}$ {(i.e., their movement including directions, speeds, etc.)}, and the time horizon~$T$. Moreover, when $C_1$ is large (compared with $C_M-C_1$), the expected competitive ratio approaches $e/(e-1)\approx 1.58$.

Regarding the computational complexity of Alg.~\ref{online-alg1}, we note that Line~\ref{online-alg1:max} takes $O(\log N)$ for the minimum search. Moreover, following \cite{tseng2019online}, we can show that there are at most $\sqrt{C_1}$ iterations such that the condition in Line~\ref{primal-dual-alg1:condition} holds; that is, we can revise the iteration in Line~\ref{primal-dual-alg1:condition} to start from $j=\max\{t-\lfloor \sqrt{C_1} \rfloor, 1\}$. Thus, the computational complexity of Alg.~\ref{online-alg1} for each slot is $\max\{O(\log N), O(\sqrt{C_1})\}$.


\subsection{Channel-Agnostic Scheduling Algorithm} \label{subsection:channe-agnostic}
This section develops a channel-agnostic scheduling algorithm by modifying Algs.~\ref{primal-dual-alg1} and~\ref{online-alg1}. Instead of searching for the minimum power level such that all users can receive the update (in Line~\ref{primal-dual-alg1:max} of Alg.~\ref{primal-dual-alg1} and Line~\ref{online-alg1:max} of Alg.~\ref{online-alg1}), the channel-agnostic scheduling algorithm modifies both lines by setting $k^*_t \leftarrow M$ for all slot~$t$. Moreover, the channel-agnostic scheduling algorithm modifies Line~\ref{online-alg1:tx} of Alg.~\ref{online-alg1} to $d(t)\leftarrow M$. This algorithm increases the transmission probability more slowly than Alg.~\ref{online-alg1} does in Line~\ref{online-alg1:x}, but always allocates the maximum power level when deciding to transmit. With these modifications, the channel-agnostic scheduling algorithm needs $C_1$ and $C_M$ only, with no channel state information. Following the proofs of Theorems~\ref{theroem:competitive-ratio1} and~\ref{theorem:online-alg1}  (with minor modifications, e.g., modify $C_{k^*_t}$ to $C_M$ in Eq.~(\ref{eq:competitive-proof})), we can show that the channel-agnostic scheduling can achieve the same competitive ratio as Alg.~\ref{online-alg1} does in Theorem~\ref{theorem:online-alg1}. Moreover, the computational complexity of the channel-agnostic scheduling algorithm for each slot is $O(\sqrt{C_1})$, independent of the number of users.

\section{Numerical Results}
While the previous section analyzes the proposed algorithms in the adversarial (worst-case) scenario, this section further validates them in stochastic scenarios. We run our algorithms under the 4-state Markov-modulated channel in Fig.~\ref{fig:channel}-(a) (with random initial states) for 10,000 slots. Moreover, we set the transmission cost to $C_k=C_1+5(k-1)$ for $k=1, \cdots, 4$.
Fig.~\ref{fig:channel}-(b) shows the ratio between the total cost incurred by the proposed algorithms and the minimum total cost for various values of $C_1$ (fixed $N=2$). The minimum total cost is obtained by an optimal offline scheduling algorithm as follows. We model our scheduling problem in the Markov-modulated channel as a Markov decision process (MDP) like~\cite{hsu2019scheduling}. Then, we identify an optimal scheduling algorithm by the value iteration algorithm. From Fig.~\ref{fig:channel}-(b), we can observe that in the stochastic scenarios our algorithms perform much better than what we analyzed in Theorem~\ref{theorem:online-alg1}.

Next,  Fig.~\ref{fig:sim1} simulates the time-average total cost $(\sum_{t=1}^T (C_{d(t)}+\sum_{i=1}^N a_i(t)/N))/T$ and the time-average age $(\sum_{t=1}^T (\sum_{i=1}^N a_i(t)/N))/T$ for various values of $C_1$ (fixed $N=5$), and  Fig.~\ref{fig:sim2} does  for various values of $N$ (fixed $C_1=30$). Because of the curse of dimensionality of MDPs, in the setting we compare the proposed algorithms with two online greedy algorithms, Greedy~1 and Greedy~2. Greedy~1 chooses $d(t)$ to minimize the total cost $C_{d(t)}+(\sum_{i=1}^N a_i(t)/N)$ at the end of slot~$t$. Let $g_i(t)$ be the \textit{cumulative} age cost for user~$i$ from the slot when it receives the previous update until the present slot~$t$, that is, $g_i(t)=0$ if user~$i$ can receive the update in slot~$t$; otherwise, $g_i(t)=g_i(t-1)+a_i(t)$. Greedy~2 chooses $d(t)$ to minimize $C_{d(t)}+(\sum_{i=1}^N g_i(t)/N)$ at the end of slot~$t$. We can observe that Greedy~2 outperforms Greedy~1. This is because Greedy~1 neglects the cost incurred by the stale information in the previous slots.  We can also observe that our algorithms significantly outperform the greedy algorithms in both cost (in Figs.~\ref{fig:sim1}-(a) and~\ref{fig:sim2}-(a)) and age (in Figs.~\ref{fig:sim1}-(b) and~\ref{fig:sim2}-(b)). Moreover, the online scheduling algorithm only has a marginal improvement over the channel-agnostic scheduling algorithm,  even though it has the present channel states of all users.

\begin{figure}[t]
		\begin{minipage}{.23\textwidth}
		\centering
		\vspace{1.8cm}
	\includegraphics[width=\textwidth]{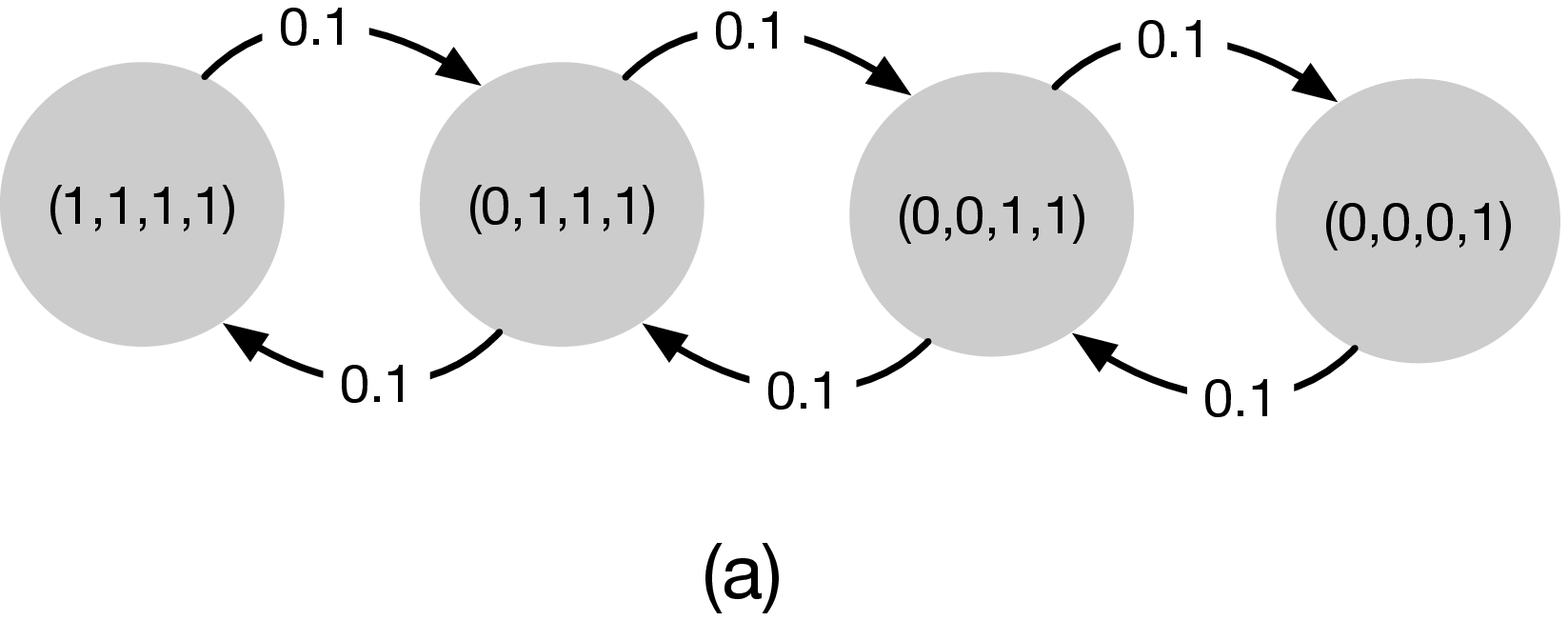}
	\end{minipage}\hfill
	\begin{minipage}{.23\textwidth}
	\centering
	\includegraphics[width=\textwidth]{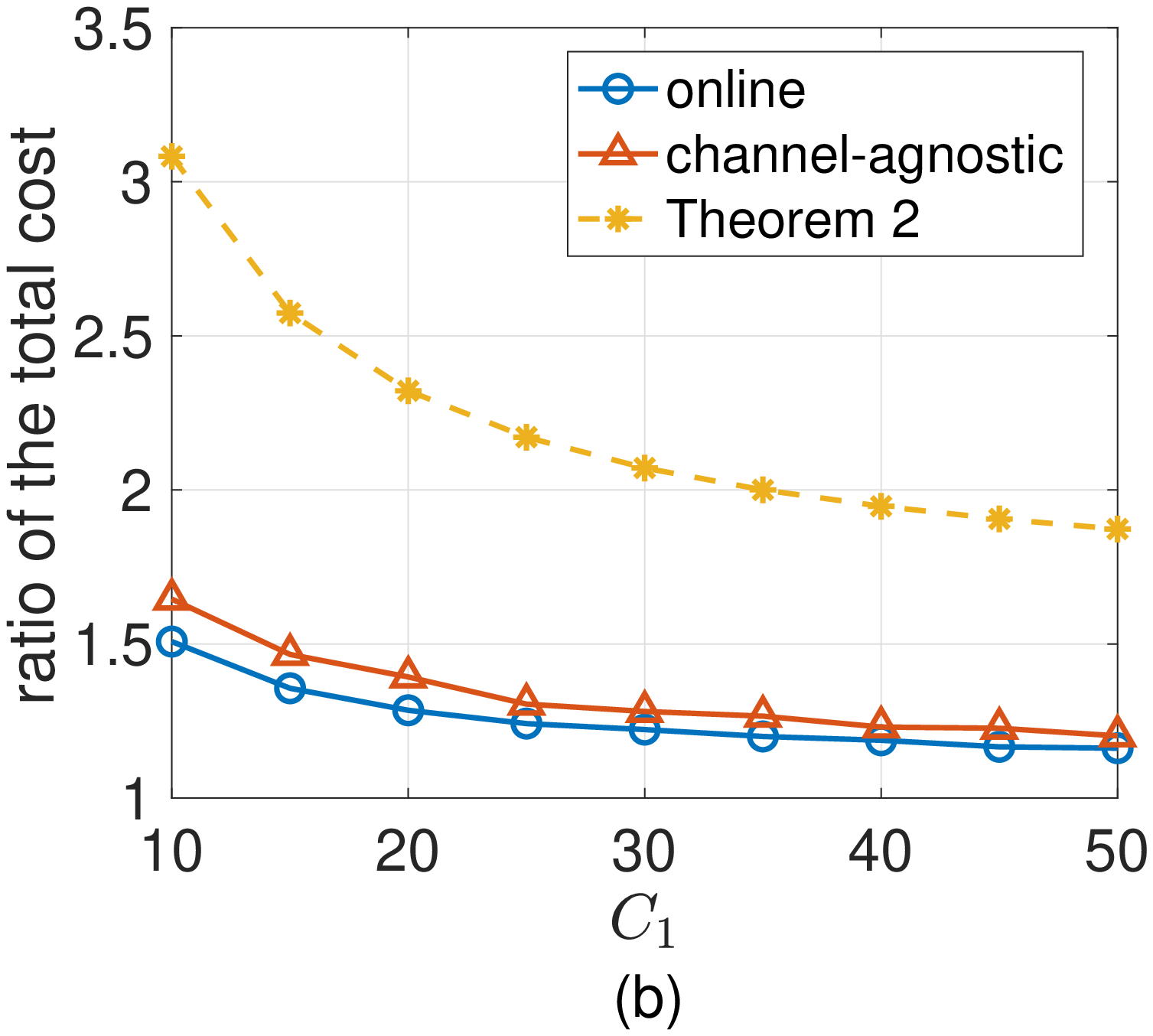}
\end{minipage}\hfill
	\caption{(a) 4-state Markov-modulated channel for $s_i(t)$; (b) ratio to the minimum total cost for $N=2$.}
	\label{fig:channel}
\end{figure}

\begin{figure}[t]
	\begin{minipage}{.23\textwidth}
		\centering
		\includegraphics[width=\textwidth]{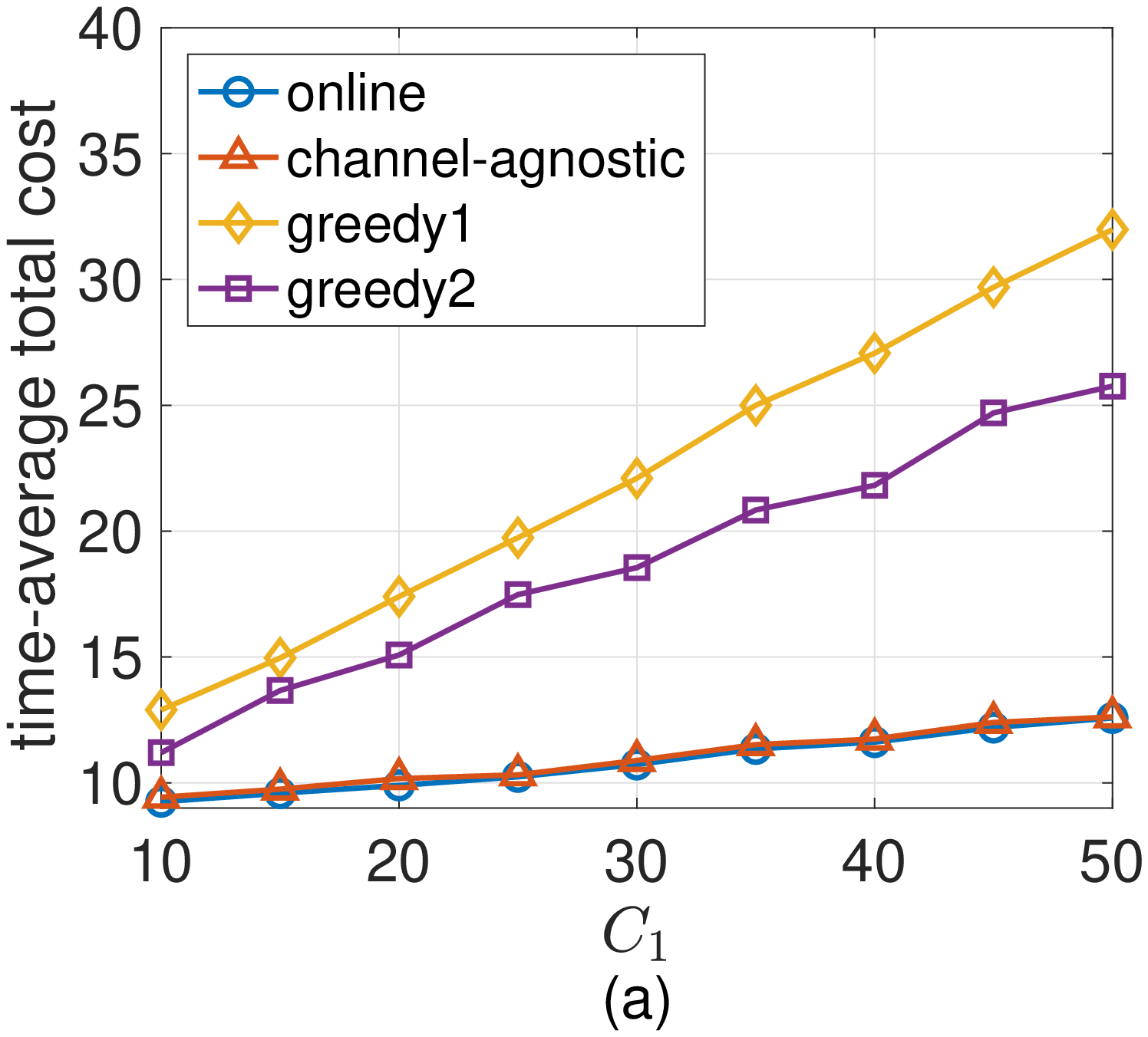}
	\end{minipage}\hfill
	\begin{minipage}{.23\textwidth}
		\centering
		\includegraphics[width=\textwidth]{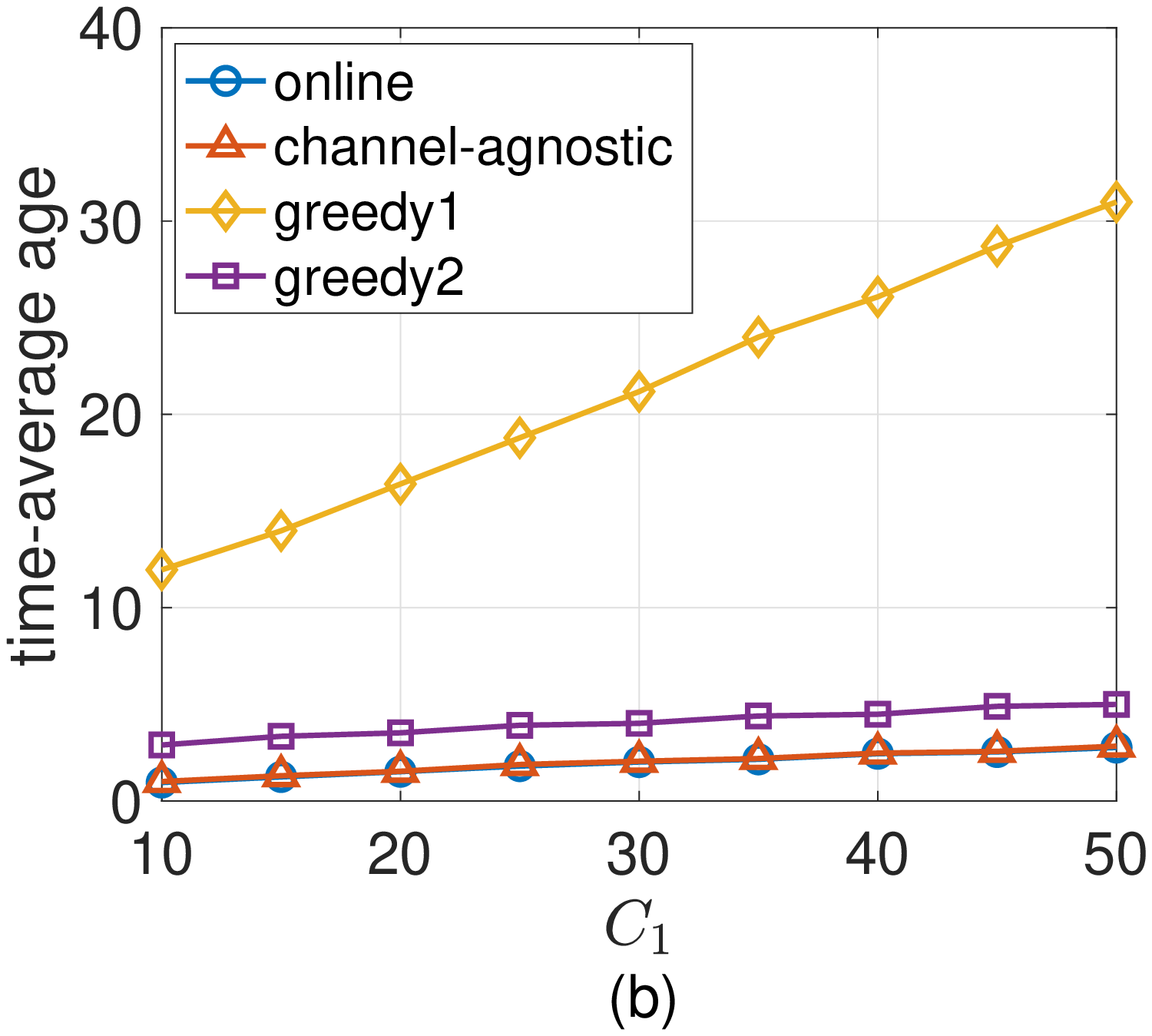}
	\end{minipage} \hfill
	\caption{(a) Time-average total cost and (b) time-average age for $N=5$.}
	\label{fig:sim1}
\end{figure}
\begin{figure}[t]
	\begin{minipage}{.24\textwidth}
		\centering
		\includegraphics[width=\textwidth]{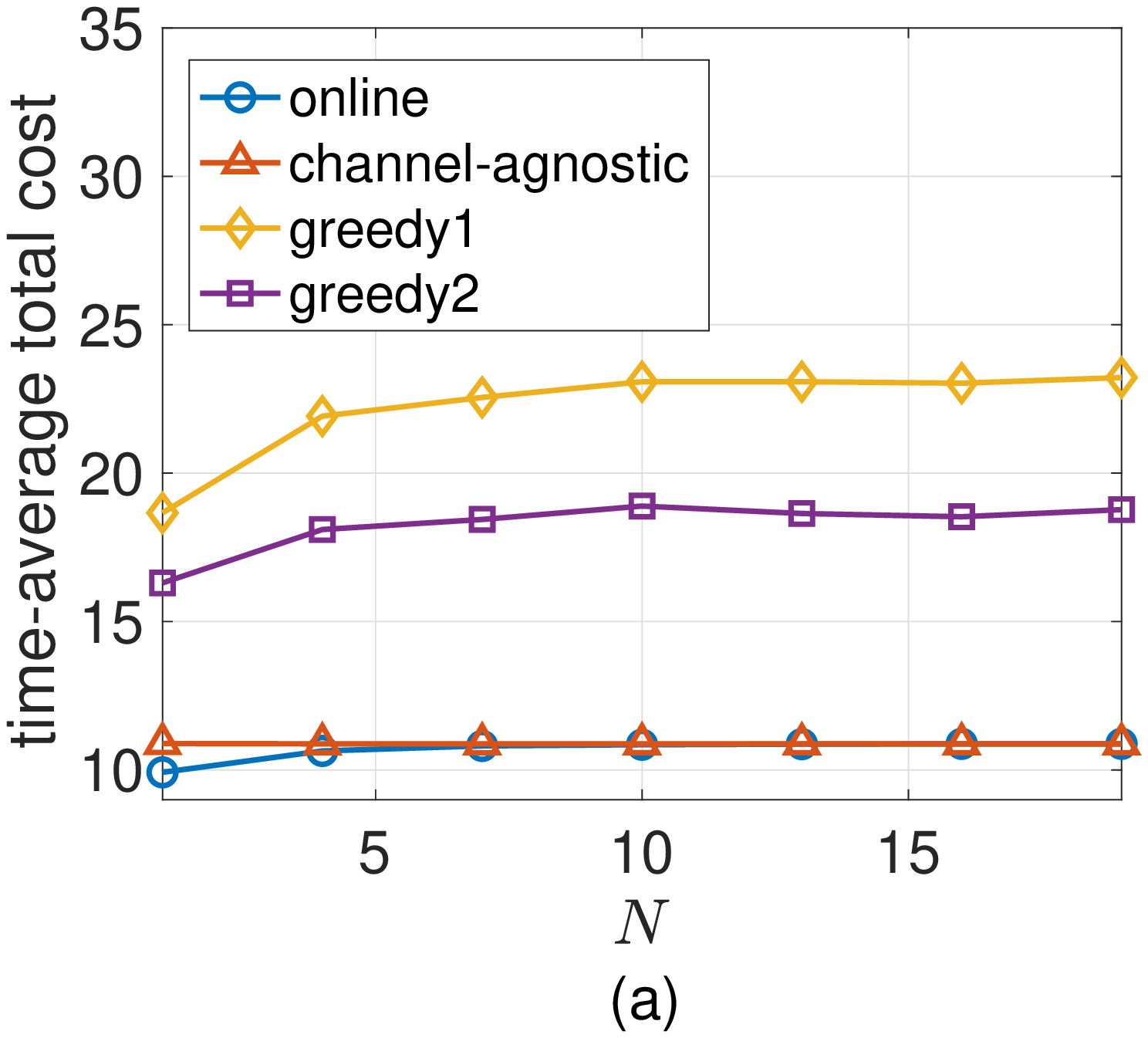}
	\end{minipage}\hfill
	\begin{minipage}{.24\textwidth}
		\centering
		\includegraphics[width=\textwidth]{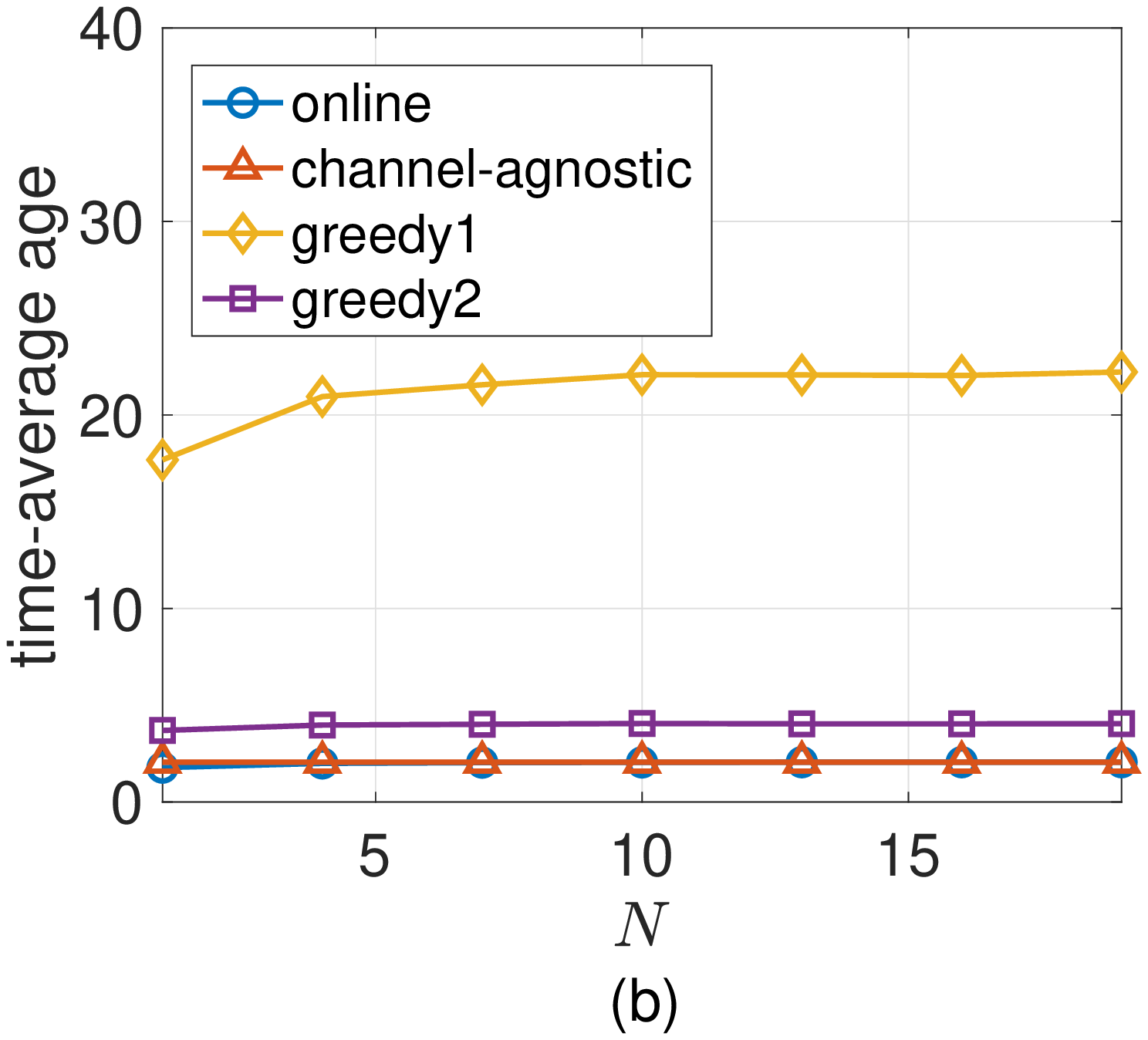}
	\end{minipage} \hfill
	\caption{(a) Time-average total cost and (b) time-average age for  $C_1=30$.}
	\label{fig:sim2}
\end{figure}

\section{Conclusion}
This study developed two scheduling algorithms to strike a balance between the timeliness of information (for the users) and the power consumption (for the AP). The proposed algorithms can achieve a constant competitive ratio, independent of the number of users and their movement. While this study focused on the competitive analysis, it is an interesting extension to explicitly derive the resulting total cost.

\appendices
\section{Proof of Theorem~\ref{theroem:competitive-ratio1}}\label{appendix:theroem:competitive-ratio1}

The proof needs the following  technical lemma.
\begin{lemma}
	Alg.~\ref{primal-dual-alg1} produces a feasible solution to the primal program and the dual program.
\end{lemma}
\begin{proof}
	The solution produced by Alg.~\ref{primal-dual-alg1} satisfies the constraints in Eq.~(\ref{interger-program:const1}) according to Line~\ref{primal-dual-alg1:z1}. Moreover,  the solution produced by Alg.~\ref{primal-dual-alg1} also satisfies the constraints in Eq.~(\ref{dual-program:const1}) because, for each power level~$k$ and slot~$t$, we can obtain
	\begin{align*}
		&\sum_{i=1}^N\mathbf{1}_{i,k}(t) \sum_{j=1}^t\sum_{\tau=t}^{T} y_{i,j}(\tau)\\
		\leq &\sum_{i=1}^N \sum_{j=1}^t\sum_{\tau=t}^{T} y_{i,j}(\tau)\\		
		\mathop{=}^{(a)}&\sum_{j=1}^t \sum_{\tau=t}^{T} \left(\sum_{i=1}^N \frac{1}{N}\right)\cdot \mathbf{1}_{\text{update in iteration~$j$ of slot $\tau$}}\\
		=&\sum_{j=1}^t \sum_{\tau=t}^{T}  \mathbf{1}_{\text{update in iteration~$j$  of slot $\tau$}}\\
		\mathop{\leq}^{(b)} & \lfloor C_1 \rfloor \leq C_k,
	\end{align*}
	where $\mathbf{1}_{\text{update in iteration~$j$ of slot $\tau$}}$  indicates if the condition in Line~\ref{primal-dual-alg1:condition} holds for iteration~$j$ of slot~$t$ (such that $y_{i,j}(\tau)$ gets updated). The equality in (a) is according to Line~\ref{primal-dual-alg1:y1}. Moreover, the inequality in (b) follows the lines in \cite{tseng2019online}.
	
	Finally, the solution produced by Alg.~\ref{primal-dual-alg1}  also satisfies the constraints in Eq.~(\ref{dual-program:const2}) according to Line~\ref{primal-dual-alg1:y1}.
\end{proof}
Then, let $\Delta \mathscr{P}_j(t)$ be the increment of  the primal objective value in Eq.~(\ref{inter-program:objective}) caused by  Alg.~\ref{primal-dual-alg1}'s  iteration~$j$ of slot~$t$. First, if the condition in Line~\ref{primal-dual-alg1:condition} of Alg.~\ref{primal-dual-alg1} holds in iteration~$j$ of slot~$t$, then $\Delta \mathscr{P}_j(t)$ becomes 
\begin{align} 
	\Delta \mathscr{P}_j(t)=&C_{k^*_{t}}(t)\underbrace{\left(\frac{1}{C_{k^*_t}} \sum_{\tau=i}^{t} x_{k^*_{\tau}}(\tau)+\frac{1}{\theta\cdot C_{k^*_t}}\right)}_{(a)}\nonumber\\
	&+\frac{1}{N}\underbrace{\left(1- \sum_{\tau=i}^{t}x_{k^*_{\tau}}(\tau)\right)N}_{(b)}
	= 1+\frac{1}{\theta},\label{eq:competitive-proof}
\end{align} 
where (a) is according to Line~\ref{primal-dual-alg1:x}; (b) is according to Line~\ref{primal-dual-alg1:z1}. Second, if the condition in Line~\ref{primal-dual-alg1:condition} fails in iteration~$j$ of slot~$t$, then $\Delta \mathscr{P}_j(t)=0$. 

Similarly, let  $\Delta \mathscr{D}_j(t)$ be the increment of   the dual objective value in Eq.~(\ref{dual-program:objective}) caused by Alg.~\ref{primal-dual-alg1}'s  iteration~$j$ of slot~$t$. First, if the condition in Line~\ref{primal-dual-alg1:condition} holds in iteration~$j$ of slot~$t$, then $\Delta \mathscr{D}_j(t)=(1/N)\cdot	N=1$ according to  Line~\ref{primal-dual-alg1:y1}. Second, if the condition in Line~\ref{primal-dual-alg1:condition} fails in   iteration~$j$ of slot~$t$, then $\Delta \mathscr{D}_j(t)=0$. Thus, we can establish  that $\Delta \mathscr{P}_j(t)=\left(1+\frac{1}{\theta} \right) \mathscr{D}_j(t)$
for all~$j$ and~$t$. 

Let $\mathscr{P}$ and $\mathscr{D}$  be the primal objective value and the dual objective value computed by Alg.~\ref{primal-dual-alg1}, respectively. Then,  $\mathscr{P}=\sum_{t=1}^T\sum_{j=1}^t \Delta \mathscr{P}_j(t)$ and $\mathscr{D}=\sum_{t=1}^T\sum_{j=1}^t \Delta \mathscr{D}_j(t)$; moreover, 
\begin{align*}
	\mathscr{P}= \left(1+\frac{1}{\theta}\right) \mathscr{D} \mathop{\leq}^{(a)}  \left(1+\frac{1}{\theta}\right) OPT(\mathbf{s}),
\end{align*}
where (a) is from the weak duality theory \cite{online-compuatation-naor}. Substituting $\theta$ in the above inequality by its definition yields the theorem.  

\section{Proof of Theorem~\ref{theorem:online-alg1}}\label{appendix:theorem:online-alg1}
Following the proof of \cite[Theorem 6]{tseng2019online}, we can obtain that the expected transmission cost incurred by Alg.~\ref{online-alg1} in slot~$t$ is less than or equal to the value of $C_{k^*_t}\cdot x_{k^*_t}(t)$ in Eq.~(\ref{inter-program:objective}) computed by Alg.~\ref{primal-dual-alg1}; moreover, we  can also obtain that the expected age for user~$i$ in slot~$t$ under Alg.~\ref{online-alg1} is less than or equal to the value of $\sum_{j=1}^t z_{i,j}(t)$ in Eq.~(\ref{inter-program:objective}) computed by Alg.~\ref{primal-dual-alg1}. Thus, the expected total cost  incurred by Alg.~\ref{online-alg1} is  less than or equal to the primal objective value computed by Alg.~\ref{primal-dual-alg1}. Then, the result immediately follows from Theorem~\ref{theroem:competitive-ratio1}.

{\small
	\bibliographystyle{IEEEtran}
	\bibliography{IEEEabrv,ref}
}

\end{document}